\documentclass[]{article}
\usepackage{amsfonts}
\usepackage{amssymb}
\usepackage[letterpaper, left=2.5cm, right=2.5cm, top=2.5cm,
bottom=2.5cm,dvips]{geometry}
\usepackage{verbatim}

\usepackage[T1]{fontenc}
\usepackage{graphicx}
\usepackage{diagbox}
\usepackage{xcolor}
\usepackage{amsmath}
\usepackage{array,multirow}
\usepackage{dsfont}
\usepackage{amsthm}
\usepackage{thmtools}
\usepackage{thm-restate}
\usepackage{enumitem}
\usepackage{bm}

\usepackage{physics}
\usepackage{mathrsfs}
\usepackage{extarrows}
\usepackage{mathdots}
\usepackage{mathtools}

\usepackage[pagebackref]{hyperref}
\hypersetup{
	colorlinks,
	linkcolor={blue!60!black},
	citecolor={green!60!black},
	urlcolor={green!60!black}
}

\definecolor{cream}{RGB}{203, 237, 204}

\newtheorem{theorem}{Theorem}[section]
\newtheorem{lemma}[theorem]{Lemma}
\newtheorem{proposition}[theorem]{Proposition}

\newtheorem{claim}[theorem]{Claim}

\numberwithin{equation}{section}

\makeatletter
\@namedef{theHequation}{\thesection.\arabic{equation}}
\newcommand{\settheoremHcounter}[1]{\@namedef{theH#1}{\thesection.\arabic{theorem}}}
\settheoremHcounter{theorem}
\settheoremHcounter{lemma}
\settheoremHcounter{proposition}
\settheoremHcounter{corollary}
\settheoremHcounter{conjecture}
\settheoremHcounter{definition}
\settheoremHcounter{construction}
\settheoremHcounter{remark}
\settheoremHcounter{claim}
\settheoremHcounter{example}
\settheoremHcounter{question}
\settheoremHcounter{problem}
\settheoremHcounter{fact}
\settheoremHcounter{observation}
\settheoremHcounter{thm}
\settheoremHcounter{prop}
\settheoremHcounter{rmk}
\settheoremHcounter{cor}
\settheoremHcounter{eg}
\settheoremHcounter{nota}
\makeatother

\newtheorem{thm}[theorem]{Theorem}

\DeclareMathOperator{\supp}{supp}
\DeclareMathOperator{\wt}{wt}

\newcommand{\R}{\mathbb{R}}

\newcommand{\C}{\mathbb{C}}
\newcommand{\cq}{\C^q}
\newcommand{\cqn}{(\cq)^{\otimes n}}
\newcommand{\kp}{\ket{\psi}}
\newcommand{\rk}{\mathrm{rank}\,}
\newcommand{\yp}{\mathbf{y}_{\mathrm p}}
\newcommand{\typ}{\tilde{\mathbf{y}}_{\mathrm p}}
\newcommand{\trans}{^{\mathsf T}}

\title{An Explicit Scott-Type Bound for Absolutely Maximally Entangled States with Arbitrary Defect}
\author{
Shixuan Zeng\thanks{School of the Gifted Young, University of Science and Technology of China.
\texttt{zsx0515@mail.ustc.edu.cn}}
\and
Xiande Zhang\thanks{School of Mathematical Sciences, University of Science and Technology of China;
Hefei National Laboratory.
\texttt{drzhangx@ustc.edu.cn}}
}
\date{}

\begin{document}

\maketitle

\begin{abstract}
Absolutely maximally entangled (AME) states and, more generally, $k$-uniform states in
$(\C^q)^{\otimes n}$ are central objects in multipartite entanglement theory, with applications
to quantum secret sharing, quantum masking, and quantum error correction. In the extremal case $k=\lfloor n/2\rfloor$, 
Scott (2004) proved a sharp nonexistence bound showing that AME states cannot exist once the number of parties $n$
exceeds a threshold of order $2q^{2}$ (with a parity dependence on $n$), where $q$ is the local dimension. Recently,
Ning et al.\  studied \emph{defective} AME states (i.e., $k=\lfloor n/2\rfloor-l$ with $l>0$), gave explicit
Scott-type bounds for defects $l=1,2$ and conjectured a general $(2l+2)q^{2}+o(q^{2})$ behavior. In this paper,  we solve this conjecture and
establish a fully explicit Scott-type upper bound for AME states with arbitrary defect $l\ge 0$, yielding Scott's bound for $l=0$ and Ning et al.'s bounds for $l=1,2$ as special cases.
Equivalently, this gives nonexistence bounds for one-dimensional pure quantum error-correcting codes near the quantum Singleton regime. 
The proof uses a truncated MacWilliams linear-programming system and an explicit infeasibility certificate.
As a direct application, we derive explicit asymptotic upper bounds on $k/n$ for fixed local dimension $q$, improving the implicit upper bounds given by Ning et al.

\end{abstract}

\section{Introduction}
Multipartite entanglement provides a structural resource in quantum information theory. In this work, we study \emph{$k$-uniform states} in $(\C^q)^{\otimes n}$, namely pure states whose $k$-body reduced density matrices are maximally mixed~\cite{GoyenecheZyczkowski2014,Scott2004}. These states also admit a coding-theoretic interpretation: a $k$-uniform state is equivalently a one-dimensional pure quantum error-correcting code with distance $k+1$~\cite{Scott2004,CalderbankRainsShorSloane1998}. Thus, the existence or nonexistence of $k$-uniform states can be viewed as the existence or nonexistence of pure quantum codes of dimension one. Uniform states also appear in quantum information tasks such as secret sharing~\cite{HelwigCuiRieraLatorreLo2012} and masking~\cite{ShiLiChenZhang2021Masking}.

The existence problem for $k$-uniform states with prescribed parameters is central to this topic \cite{FengJinXingYuan2017,GoyenecheAlsinaLatorreRieraZyczkowski2015,GoyenecheRaissiDiMartinoZyczkowski2018,GoyenecheZyczkowski2014,HuberEltschkaSiewertGuhne2018,HuberGuhneSiewert2017,LiWang2019,Rains1999ShadowEnum,RaissiTeixidoGogolinAcin2020,ShiNingZhaoZhang2025}. A particularly important case is the extremal regime in which $k$ reaches half of the number of parties; such states are called absolutely maximally entangled (AME) states.
Orthogonal arrays and classical error-correcting codes give systematic families of $k$-uniform states~\cite{FengJinXingYuan2017,GoyenecheZyczkowski2014,LiWang2019}. AME states have also been connected with combinatorial designs and multiunitary matrices~\cite{GoyenecheAlsinaLatorreRieraZyczkowski2015}, while quantum combinatorial designs provide further constructions of $k$-uniform states~\cite{GoyenecheRaissiDiMartinoZyczkowski2018}. More recent constructions go beyond the maximum-distance-separable-code framework and produce additional examples of $k$-uniform and AME states~\cite{RaissiTeixidoGogolinAcin2020}. On the other hand, nonexistence results have been developed using quantum weight enumerators, shadow enumerators, and related linear-programming constraints~\cite{HuberEltschkaSiewertGuhne2018,Rains1999ShadowEnum,ShiNingZhaoZhang2025}; for instance, the seven-qubit AME state was ruled out in~\cite{HuberGuhneSiewert2017}. In this line, deriving necessary conditions, or equivalently upper bounds on $k$ in terms of the number of systems $n$, is particularly important~\cite{HuberEltschkaSiewertGuhne2018,ShiNingZhaoZhang2025}. 

A basic necessary condition for the existence of $k$-uniform states in $(\C^q)^{\otimes n}$ is $k\le \lfloor n/2\rfloor,$
which follows either from the quantum Singleton bound~\cite{AshikhminLitsyn1999,GrasslHuberWinter2022} or directly from the Schmidt-rank argument.
The extremal case $k=\lfloor n/2\rfloor$ corresponds to AME states, which represent the strongest form of multipartite entanglement among all pure states. However,  AME states do not exist in general. For example, AME states in $(\C^2)^{\otimes n}$ exist if and only if $n =
2,3,5,6$ \cite{Rains1998Shadow,Scott2004,HuberGuhneSiewert2017,HuberEltschkaSiewertGuhne2018}. For larger dimension
$q$, many cases remain open, and we refer to the online table of AME states
maintained by Huber and Wyderka \cite{Huber:AMEtable}. 
Scott~\cite{Scott2004} proved the nonexistence of AME states in the following celebrated parity-dependent bound~\cite{Scott2004} in 2004, which rules out all AME states for large $n$ compared to the dimension $q$.

\begin{thm}[Scott's bound~\cite{Scott2004}]\label{thm:scott}
AME states in $(\C^q)^{\otimes n}$ do not exist if
\[
n>
\begin{cases}
2q^{2}-2, & n \ \text{even},\\
2q^{2}+2q-1, & n \ \text{odd}.
\end{cases}
\]
\end{thm}

Beyond the AME regime, it is natural to study states that are close to AME states, that is, when $k$ is close to
$\lfloor n/2\rfloor$. Following the notation in~\cite{NingShiLuoZhang2025}, for an integer $l\ge 0$, we call an $(\lfloor n/2\rfloor-l)$-uniform
state an \emph{AME state of defect $l$}. 
The defect parameter quantifies the distance to the AME threshold
and provides a convenient scale to formulate refined Scott-type bounds.
For qubit systems, Rains' shadow bounds \cite{Rains1998Shadow,Rains1999ShadowEnum} and their improvements \cite{Han2006Nonexistence,Han2009Nonexistence} showed the nonexistence of AME states of defect approximately $n/6$. For qutrit systems, Shi \emph{et al.} \cite{ShiNingZhaoZhang2025} extended the shadow bounds and showed the nonexistence of AME states of defect approximately $n/14$.
There are no such systematic results for general dimension $q$.
Recently,  Ning et al. \cite{NingShiLuoZhang2025} obtained explicit
Scott-type bounds for defects $l=1,2$ for general $q$ and proposed a general asymptotic picture for fixed $l$.

\begin{thm}[Ning et al.\ \cite{NingShiLuoZhang2025}]\label{thm:ning}
Let $q\ge 2$.
\begin{enumerate}
\item[(i)] AME states of defect $1$ in $(\C^q)^{\otimes n}$ do not exist if
\[
n>
\begin{cases}
4q^{2}, & n \ \text{even},\\
4q^{2}+4q+1, & n \ \text{odd}.
\end{cases}
\]
\item[(ii)] AME states of defect $2$ in $(\C^q)^{\otimes n}$ do not exist if
\[
n>
\begin{cases}
6q^{2}+2, & n \ \text{even},\\
6q^{2}+6q+3, & n \ \text{odd}.
\end{cases}
\]
\end{enumerate}
\end{thm}

Based on Theorem~\ref{thm:ning}, Ning et al.\ \cite{NingShiLuoZhang2025} further conjectured that for each fixed defect $l$ there should exist a
Scott-type nonexistence threshold with leading term $(2l+2)q^{2}$ (up to lower-order corrections)~\cite{NingShiLuoZhang2025}. Our first result 
confirms the predicted leading behavior and, moreover, provides a closed-form bound
with explicit lower-order terms and parity dependence. The result is stated below,  which  settles the Scott-type bound problem for \emph{arbitrary} defect.

\begin{thm}\label{thmmain}
Let $q\ge 2$ and $l\in\mathbb Z_{\ge0}$. Then AME states of defect $l$ in $(\C^q)^{\otimes n}$ do not exist if
\[
n>
\begin{cases}
2(l+1)q^{2}+2(l-1), & n \ \text{even},\\
2(l+1)q^{2}+2(l+1)q+(2l-1), & n \ \text{odd}.
\end{cases}
\]
\end{thm}
When $l=0$, Theorem~\ref{thmmain} recovers Scott's bound in Theorem~\ref{thm:scott}; when $l=1,2$, it recovers the bounds of Ning et al.\ in Theorem~\ref{thm:ning}. 
 For $q=2,3$, Theorem~\ref{thmmain} does not reach the shadow-bound regimes for defects of order approximately $n/6$ and $n/14$, respectively. However, it gives the first explicit Scott-type bound for arbitrary defect $l$ and arbitrary local dimension $q$.

 By the standard correspondence between $k$-uniform states and one-dimensional pure quantum error-correcting codes,  the defect parameter also measures the deficiency from the quantum Singleton extremal regime.
 Then Theorem~\ref{thmmain} can  be read as a length bound for pure quantum codes $((n,1,d))_q$ near the quantum Singleton regime. Indeed, the associated code distance is $d=k+1$, and an AME state of defect $l$ corresponds to
\[
d=\left\lfloor\frac n2\right\rfloor-l+1.
\]
Thus  Theorem~\ref{thmmain} rules out one-dimensional pure quantum codes with this distance whenever $n$ exceeds the same parity-dependent threshold.

Further, as a consequence of Theorem~\ref{thmmain}, we   derive an explicit  asymptotic upper bound on $k/n$ for $k$-uniform states at fixed local dimension $q$.
 \begin{thm}\label{thm:asymptotic_kn_bound_intro}
    Fix an integer $q\ge 2$. Suppose that there exists a $k$-uniform state in $(\C^q)^{\otimes n}$. Then
    \[
    \frac{k}{n}\le\begin{cases}
        \frac{q^2}{2(q^2+1)}+O(\frac{1}{n}), & n \text{ even},\\
        \frac{q^2+q}{2(q^2+q+1)}+O(\frac{1}{n}), & n \text{ odd}.
    \end{cases}
    \]
    \end{thm}

Ning et al.\ \cite[Theorem IV.5]{NingShiLuoZhang2025} also derived asymptotic upper bounds on $k/n$ for fixed $q$ through an implicit logarithmic inequality, reporting numerical constants for $4\le q\le 9$.
We will show that for every $q\ge4$, the constants in Theorem~\ref{thm:asymptotic_kn_bound_intro} are strictly smaller
than the bound obtainable from that asymptotic criterion in \cite{NingShiLuoZhang2025}, thereby
providing an analytic improvement. This result will be given in Proposition~\ref{prop:strict_improvement_ning}.  For better clarity, in Table~\ref{tab:compare}, we compare our explicit constants with the numerical constants for $q\in\{4,5,6,7,8,9\}$ reported by Ning et al.,  see Table~IV in \cite{NingShiLuoZhang2025}.


    
    \begin{table}[htbp]
        \centering
        \begin{tabular}{cc|cccccc}
        \hline
        \multicolumn{2}{c|}{$q$} & 4 & 5 & 6 & 7 & 8 & 9 \\
        \hline
        \multicolumn{2}{c|}{Ning et al. \cite{NingShiLuoZhang2025}}
        & 0.479 & 0.487 & 0.491 & 0.494 & 0.495 & 0.496 \\
        \hline
        \multirow{2}{*}{This paper}
        & even $n$
        & 0.471 & 0.481 & 0.486 & 0.490 & 0.492 & 0.494 \\
        & odd $n$
        & 0.476 & 0.484 & 0.488 & 0.491 & 0.493 & 0.495 \\
        \hline
        \end{tabular}
        \caption{Comparison of asymptotic upper bounds on $k/n$ for fixed $4\leq q\leq 9$.
        The numerical bounds of Ning et al.\  \cite{NingShiLuoZhang2025} are listed in the second row.
        The last two rows give our explicit bounds from Theorem~\ref{thm:asymptotic_kn_bound_intro}, distinguished by the parity of $n$. When $q=2,3$, the best ratios are $1/3 \approx 0.333$ and $3/7 \approx 0.429$, respectively, by Rains' shadow bound \cite{Rains1998Shadow} and its qutrit extension \cite{ShiNingZhaoZhang2025}.}
        \label{tab:compare}
    \end{table}   

At a high level, our proof is an enumerator-based linear-programming argument. We impose the standard constraints satisfied by quantum weight and unitary enumerators, linked by the MacWilliams-type identities~\cite{HuberEltschkaSiewertGuhne2018,NingShiLuoZhang2025,Rains1998WeightEnum,ShorLaflamme1997}. For an AME state of defect $l$, these constraints contain a finite truncated MacWilliams subsystem involving $2l$ non-negative coefficients
$A_{k+1},A_{k+2},\dots,A_{k+2l}$.
The key contribution is an explicit infeasibility certificate for this truncated system once $n$ exceeds the claimed threshold. This certificate applies uniformly to arbitrary defect $l$, extends the previous analyses for $l=1,2$ in~\cite{NingShiLuoZhang2025}, and yields the closed-form Scott-type bound in Theorem~\ref{thmmain}.

\medskip
\noindent\textbf{Organization.}
In Section~\ref{sec:preliminary} we recall the enumerator constraints needed in the proof. Section~\ref{sec:proof}
establishes Theorem~\ref{thmmain}. Section~\ref{sec:asymptotic} derives the asymptotic bound on $k/n$
and compares it with the bounds of Ning et al.\  \cite{NingShiLuoZhang2025}.

\section{Preliminaries}\label{sec:preliminary}

\subsection{Uniform entangled states}

Throughout the paper, let $q\ge 2$ and let $\mathcal H=(\C^q)^{\otimes n}$ be the Hilbert space of an $n$-partite quantum system with local dimension $q$. For integers $a\le b$, write $[a,b]=\{a,a+1,\dots,b\}$ and $[n]=[1,n]$. A pure state in $\mathcal H$ will be denoted by $|\psi\rangle$, and its density matrix is $\rho=|\psi\rangle\langle\psi|$. For a subset $S\subseteq[n]$, we write $S^c=[n]\setminus S$ and define the reduced density matrix on the parties in $S$ by $\rho_S=\operatorname{Tr}_{S^c}(\rho)$, where $\operatorname{Tr}_{S^c}$ denotes the partial trace over the tensor factors indexed by $S^c$. Thus $\rho_S$ is a positive semidefinite matrix of trace one acting on $(\C^q)^{\otimes |S|}$. For standard background on density operators, partial traces, and reduced states, we refer to \cite{NielsenChuang2010,Watrous2018}.

A state $|\psi\rangle$ is called \emph{$k$-uniform} if every $k$-party reduced density matrix is maximally mixed, namely
\[
\rho_S=\frac{I_{q^k}}{q^k},\qquad \forall\, S\subseteq[n],\ |S|=k.
\]
Equivalently, every reduced density matrix on at most $k$ parties is maximally mixed. Indeed, if $T\subseteq[n]$ and $|T|\le k$, then one can choose $S\supseteq T$ with $|S|=k$, and tracing $\rho_S=I_{q^k}/q^k$ over $S\setminus T$ gives $\rho_T=I_{q^{|T|}}/q^{|T|}$. In this sense, $k$-uniformity says that no subsystem of size at most $k$ contains any local information about the global pure state.

The largest possible value of $k$ is $\lfloor n/2\rfloor$. To see this, consider a bipartition $S\sqcup S^c=[n]$. Since $\rho$ is pure, the two reduced density matrices $\rho_S$ and $\rho_{S^c}$ have the same nonzero eigenvalues, counted with multiplicity. If $\rho_S$ is maximally mixed, then $\rho_S$ has rank $q^{|S|}$. Hence $q^{|S|}\le q^{n-|S|}$, and therefore $|S|\le n/2$. The extremal case $k=\lfloor n/2\rfloor$ is called an \emph{absolutely maximally entangled} state, or an \emph{AME state}.

We will also use the standard coding-theoretic interpretation. A $k$-uniform state in $(\C^q)^{\otimes n}$ is equivalently a one-dimensional pure quantum error-correcting code with parameters $((n,1,k+1))_q$; in particular, its code distance is $d=k+1$~\cite{Scott2004,CalderbankRainsShorSloane1998}. More generally, following Ning et al.~\cite{NingShiLuoZhang2025}, if $k=\lfloor n/2\rfloor-l$ for some integer $l\ge0$, we call a $k$-uniform state an \emph{AME state of defect $l$}. In code language, this corresponds to
\[
d=\left\lfloor\frac n2\right\rfloor-l+1.
\]
Thus the defect parameter measures the distance from the AME threshold, or equivalently the deficiency from the near-Singleton regime for one-dimensional pure quantum codes.

\subsection{Weight and unitary enumerators}

We now recall the enumerator formalism used in the proof, following the quantum weight-enumerator framework of Shor--Laflamme and Rains~\cite{ShorLaflamme1997,Rains1998WeightEnum}. Choose an orthogonal unitary error basis $\{e_a\}_{a=0}^{q^2-1}\subseteq \C^{q\times q}$ satisfying $e_0=I_q$ and $\operatorname{Tr}(e_a^\dagger e_b)=q\delta_{ab}$. For an $n$-fold tensor product $E=e_{a_1}\otimes e_{a_2}\otimes\cdots\otimes e_{a_n}$, define its weight by $\wt(E)=|\{r\in[n]:a_r\ne0\}|$. Thus the weight counts the number of tensor factors on which the error operator is nontrivial.

The \emph{weight enumerator} of the pure state $\rho=|\psi\rangle\langle\psi|$ is the polynomial $A(x,y)=\sum_{i=0}^n A_i x^{n-i}y^i$, where
\[
A_i=\sum_{\wt(E)=i}|\operatorname{Tr}(E\rho)|^2.
\]
The coefficients $A_i$ are nonnegative by definition, and $A_0=|\operatorname{Tr}(\rho)|^2=1$. One should think of $A_i$ as measuring the total strength of the $i$-body correlations of the state with respect to the chosen error basis. If $|\psi\rangle$ is $k$-uniform, then all nontrivial observables supported on at most $k$ parties have zero expectation \cite{Scott2004,NingShiLuoZhang2025}, so 
\[A_1=A_2=\cdots=A_k=0,\quad A_i\geq 0\text{ for }i\in [k+1,n].\] This is the first source of linear constraints in the proof.

The second enumerator is expressed in terms of purities of reduced density matrices. For $S\subseteq[n]$, define $A'_S=\operatorname{Tr}(\rho_S^2)$. This quantity is the purity of the reduced state on $S$. For each $i\in[0,n]$, set \[A'_i=\sum_{|S|=i}A'_S =\sum_{|S|=i} \operatorname{Tr}(\rho_S^2)\] and define the \emph{unitary enumerator} $A'(x,y)=\sum_{i=0}^n A'_i x^{n-i}y^i$. Since $\rho$ is pure, the reductions $\rho_S$ and $\rho_{S^c}$ have the same nonzero eigenvalues. Hence $\operatorname{Tr}(\rho_S^2)=\operatorname{Tr}(\rho_{S^c}^2)$, and after summing over all subsets of size $i$ we get the symmetry \[A'_i=A'_{n-i}, ~i\in [0,n].\]
For any density matrix $\sigma$ on a $d$-dimensional Hilbert space, one has $\operatorname{Tr}(\sigma^2)\ge 1/d$, with equality precisely when $\sigma=I_d/d$; this follows from Cauchy--Schwarz applied to the eigenvalues of $\sigma$ \cite{NielsenChuang2010,Watrous2018}. Hence for every $S$ with $|S|=i$, one has $\operatorname{Tr}(\rho_S^2)\ge 1/q^i$. So for any $i\in [0,n]$,
\[A'_i\ge q^{-i}\binom{n}{i}.\]
If $|\psi\rangle$ is $k$-uniform and $i\le k$, then every $i$-party reduced density matrix is maximally mixed, so $\operatorname{Tr}(\rho_S^2)=1/q^i$ for $|S|=i$. Therefore \[A'_i=q^{-i}\binom{n}{i} \text{  for $0\le i\le k$.}\] 

\subsection{Enumerator constraints}

The weight and unitary enumerators are linked by a quantum MacWilliams-type identity, originating from the enumerator theory of quantum error-correcting codes~\cite{ShorLaflamme1997,Rains1998WeightEnum,HuberEltschkaSiewertGuhne2018}. In the present rank-one setting of a pure state, the form needed in this paper is
\begin{equation}\label{eq:mcwlm}
    q^iA'_i=\sum_{j=0}^i A_j\binom{n-j}{i-j},\qquad 0\le i\le n.
\end{equation}

For completeness, we indicate why this relation holds. The reduced density matrix $\rho_S$ can be expanded in the tensor-product error basis supported inside $S$. Orthogonality of the basis gives
\[
q^{|S|}\operatorname{Tr}(\rho_S^2)
=
\sum_{\supp(E)\subseteq S}|\operatorname{Tr}(E\rho)|^2.
\]
Now sum this identity over all subsets $S\subseteq[n]$ with $|S|=i$. An error operator of weight $j$ is supported in exactly $\binom{n-j}{i-j}$ subsets $S$ of size $i$. This gives Eq. (\ref{eq:mcwlm}).

We collect all the constraints on the two kinds of enumerators in the following lemma, which will be used to deduce the contradiction in the rest of the paper.

\begin{lemma}[Enumerator constraints for $k$-uniform states \cite{NingShiLuoZhang2025}]\label{lem:enumerator_constraints}
Let $|\psi\rangle\in(\C^q)^{\otimes n}$ be a $k$-uniform state, and let
$A(x,y)=\sum_{i=0}^{n}A_{i}x^{n-i}y^{i}$ and $A'(x,y)=\sum_{i=0}^{n}A'_{i}x^{n-i}y^{i}$ be its weight and unitary
enumerators. Then
\begin{align}
&A_{0}=1,\qquad A_{1}=A_{2}=\cdots=A_{k}=0,\qquad A_{i}\ge 0,\ \ \forall\, i\in[k+1,n]; \label{eqA}\\
&A'_{i}=\frac{1}{q^{i}}\binom{n}{i},\ \ \forall\, i\in[0,k];\qquad
A'_{i}\ge \frac{1}{q^{i}}\binom{n}{i},\ \ \forall\, i\in[k+1,\lfloor n/2\rfloor];\qquad
A'_{i}=A'_{n-i},\ \ \forall\, i\in[0,n]; \label{eqB}\\
&q^{i}A'_{i}=\sum_{j=0}^{i}A_{j}\binom{n-j}{i-j},\qquad \forall\, i\in[0,n]. \label{eqC}
\end{align}
Consequently, if the constraints of Eqs.~\eqref{eqA}--\eqref{eqC} are infeasible for some given $(n,q,k)$, then no $k$-uniform state
exists in $(\C^q)^{\otimes n}$.
\end{lemma}

\subsection{Main ideas}

Lemma~\ref{lem:enumerator_constraints} translates the existence of a $k$-uniform state into a linear-programming feasibility problem over quantum enumerators. In the positive-defect part of the proof of Theorem~\ref{thmmain}, we fix $l\ge1$, equivalently $k=\lfloor n/2\rfloor-l$. Since $A_1,\dots,A_k$ vanish, the first potentially nonzero weight-enumerator coefficients are $A_{k+1},A_{k+2},\dots$. We extract from the MacWilliams-type relation~\eqref{eqC} the truncated subsystem with
\[
i=k+1,k+2,\dots,k+2l.
\]
After the known terms $A_0=1$ and $A_1=\cdots=A_k=0$ are separated, this subsystem gives a lower-triangular linear system for
\[
A_{k+1},A_{k+2},\dots,A_{k+2l}.
\]

The symmetry and lower bounds for the unitary enumerator in Eq.~\eqref{eqB} provide the remaining constraints. Thus the nonexistence problem reduces to proving infeasibility of a finite truncated MacWilliams linear program. The main point of the proof is to construct an explicit certificate which annihilates the free parameters in this truncated system and forces a negative value for a nonnegative enumerator expression. This gives a contradiction to the enumerator constraints and yields the claimed Scott-type bound.

\section{Proof of Theorem~\ref{thmmain}}\label{sec:proof}
    Suppose that a $k$-uniform state $\kp$ exists with $k=\lfloor n/2\rfloor-l$; we will derive a contradiction. The case $l=0$ is Scott's bound~\cite{Scott2004}, while $l=1,2$ is covered by Ning et al.~\cite{NingShiLuoZhang2025}. Hence we assume $l\ge 3$ from now on. Throughout the proof, we use the convention that $\binom{a}{b}=0$ whenever $b<0$ or $b>a$. Following the enumerator constraints framework recalled in Lemma~\ref{lem:enumerator_constraints}, we first isolate a finite block of Eq.~\eqref{eqC}. Put
    \[
    M\triangleq\begin{pmatrix}
        \binom{n-k-1}{0} &  &  & \\[2pt]
        \binom{n-k-1}{1} & \binom{n-k-2}{0} &  & \\[2pt]
        \vdots & \vdots & \ddots & \\[2pt]
        \binom{n-k-1}{2l-1} & \binom{n-k-2}{2l-2} & \cdots & \binom{n-k-2l}{0}
    \end{pmatrix}_{2l\times 2l}, \qquad M_{i,j}\triangleq\begin{cases}
        \binom{n-k-j}{i-j}, &i\ge j,\\
        0, &i<j.
    \end{cases}
    \]
    This lower triangular matrix will be used to collect the unknown coefficients. For $i=k+1,\dots,k+2l$, Eq.~\eqref{eqC} gives the equivalent system
    \begin{equation}\label{eqM}
        M\begin{pmatrix}
            A_{k+1} \\[2pt] A_{k+2} \\[2pt] \vdots \\[2pt] A_{k+2l}
        \end{pmatrix}=\begin{pmatrix}
            q^{k+1}A_{k+1}' \\[2pt] q^{k+2}A_{k+2}' \\[2pt] \vdots \\[2pt] q^{k+2l}A_{k+2l}'
        \end{pmatrix}-\begin{pmatrix}
            \binom{n}{k+1} \\[2pt] \binom{n}{k+2} \\[2pt] \vdots \\[2pt] \binom{n}{k+2l}
        \end{pmatrix}.
    \end{equation}
    A direct binomial inversion gives 
    \[
    M^{-1}=\begin{pmatrix}
        \binom{n-k-1}{0} &  &  & \\[2pt]
        -\binom{n-k-1}{1} & \binom{n-k-2}{0} &  & \\[2pt]
        \vdots & \vdots & \ddots & \\[2pt]
        -\binom{n-k-1}{2l-1} & \binom{n-k-2}{2l-2} & \cdots & \binom{n-k-2l}{0}
    \end{pmatrix}_{2l\times 2l}, \qquad (M^{-1})_{i,j}=\begin{cases}
        (-1)^{i+j}\binom{n-k-j}{i-j}, &i\ge j,\\
        0, &i<j.
    \end{cases}
    \]
    Write $\mathbf{v}_i$ for the $i$-th column of $M^{-1}$, where $1\le i\le 2l$.
    
    \subsection{\texorpdfstring{The Case of $n$ Even}{The Case of n Even}}
        Assume first that $n=2k+2l$. By the symmetry in Eq.~\eqref{eqB}, the last primed coefficient in Eq.~\eqref{eqM} is known:
        \[
        A_{k+2l}'=A_k'=\frac{1}{q^k}\binom{n}{k}.
        \]
        The remaining quantities
        \[
        A_{k+1},A_{k+2},\dots,A_{k+2l},A_{k+1}',A_{k+2}',\dots,A_{k+2l-1}'
        \]
        are treated as variables. 
        
        \textbf{Step 1: } We first give general forms of $A_{k+1},A_{k+2},\dots,A_{k+2l},A_{k+1}',A_{k+2}',\dots,A_{k+2l-1}'$ satisfying Eq.~\eqref{eqM}, which will be in Eq.~\eqref{sol_even} and \eqref{eq:ytilde_even}.
        Incorporating the symmetry relations among the primed variables, Eq.~\eqref{eqM} becomes
        \begin{equation}\label{eq_even}
            L_{(3l-1)\times(4l-1)}\begin{pmatrix}
                A_{k+1} \\ \vdots \\ A_{k+2l} \\ A_{k+1}' \\ \vdots \\ A_{k+2l-1}'
            \end{pmatrix}=\begin{pmatrix}
                -\binom{n}{k+1} \\ \vdots \\ -\binom{n}{k+2l-1} \\ (q^{2l}-1)\binom{n}{k+2l} \\ O_{(l-1)\times 1}
            \end{pmatrix},
        \end{equation}
        where
        \begin{equation*}
            L\triangleq\begin{pmatrix}
            M_{2l\times 2l} & P_{2l\times(2l-1)}\\
            O_{(l-1)\times 2l} & Q_{(l-1)\times (2l-1)}
            \end{pmatrix}
        \end{equation*}
        and 
        \begin{equation*}
            P\triangleq\begin{pmatrix}
                 -q^{k+1} &           &           &             &\\
              & -q^{k+2}  &           &             &\\
              &           & \ddots    &             &\\
              &           &           & -q^{k+2l-1} &\\
             0& 0         & \cdots    & 0              
            \end{pmatrix},\qquad Q\triangleq\begin{pmatrix}
            1        & \cdots & 0      & 0      & 0      & \cdots  & -1     \\
                     & \ddots & 0      & 0      & 0      & \iddots &        \\
                     &        & 1      & 0      & -1     &         &
            \end{pmatrix},
        \end{equation*} and $O$ is the all zero matrix.
        
        The invertibility of $M$, together with the displayed block form, gives $\rk L=3l-1$. Thus the solution set of Eq.~\eqref{eq_even} is an $l$-dimensional affine subspace. We parametrize its homogeneous part as follows. Define vectors in $\R^{4l-1}$ by
        \[
        \mathbf{y}_1\triangleq\begin{pmatrix}
            q^{k+1}\mathbf{v}_1 \\ 1 \\ 0 \\ \vdots \\ 0
        \end{pmatrix},\qquad \mathbf{y}_2\triangleq\begin{pmatrix}
            q^{k+2}\mathbf{v}_2 \\ 0 \\ 1 \\ \vdots \\ 0
        \end{pmatrix},\qquad\cdots,\qquad \mathbf{y}_{2l-1}\triangleq\begin{pmatrix}
            q^{k+2l-1}\mathbf{v}_{2l-1} \\ 0 \\ 0 \\ \vdots \\ 1
        \end{pmatrix}.
        \]
        A direct check gives
        \[
            L(\mathbf{y}_1+\mathbf{y}_{2l-1})=\cdots=L(\mathbf{y}_{l-1}+\mathbf{y}_{l+1})=L\mathbf{y}_l=0.
        \]
        Set
        \[
        \mathbf{w}_1\triangleq\mathbf{y}_1+\mathbf{y}_{2l-1},\qquad\cdots,\qquad \mathbf{w}_{l-1}\triangleq\mathbf{y}_{l-1}+\mathbf{y}_{l+1}, \qquad \mathbf{w}_l\triangleq\mathbf{y}_l.
        \]
        The last $2l-1$ coordinates show that the vectors in $W:=\{\mathbf{w}_i\mid1\le i\le l\}$ are linearly independent. Since the nullity is $l$, this set forms a basis of $\mathrm{Ker}\,L$.

        We next choose a particular solution of Eq.~\eqref{eq_even}. Write
        \[
        \yp=(a_{k+1},\dots,a_{k+2l},a_{k+1}',\dots,a_{k+2l-1}')\trans.
        \]
        The primed coordinates are assigned using Eq.~\eqref{eqB} and the symmetry $A_i'=A_{n-i}'$: set
        \[
        \begin{aligned}
            a_{k+1}'=a_{k+2l-1}'&=\frac{1}{q^{k+1}}\binom{n}{k+1},\\
            &\vdots\\
            a_{k+l-1}'=a_{k+l+1}'&=\frac{1}{q^{k+l-1}}\binom{n}{k+l-1},\\
            a_{k+l}'&=\frac{1}{q^{k+l}}\binom{n}{k+l}.
        \end{aligned}
        \]
        Inserting these values and $A_{k+2l}'=\frac{1}{q^k}\binom{n}{k}$ into Eq.~\eqref{eqM} gives $a_{k+1}=\dots=a_{k+l}=0$ and
        \begin{equation}\label{value_even}
            \begin{aligned}
                a_{k+l+1}&=(q^2-1)\binom{n}{k+l+1},\\
                a_{k+l+2}&=(q^4-1)\binom{n}{k+l+2}-(n-(k+l+1))(q^2-1)\binom{n}{k+l+1}.
            \end{aligned}
        \end{equation}
        The other entries $a_{k+l+i}$ are determined uniquely for $1\le i\le l$, but only $a_{k+l+1}$ and $a_{k+l+2}$ will be used below. Since $\yp$ satisfies Eq.~\eqref{eq_even}, every solution of Eq.~\eqref{eq_even} has the form
        \begin{equation}\label{sol_even}
        \mathbf{y}=t_1\mathbf{w}_1+\cdots+t_{l}\mathbf{w}_{l}+\yp, \qquad \text{ for some } t_1,\dots,t_l\in\R.
        \end{equation}
        Project $\mathbf{y}$ and $\mathbf{y}_{\mathrm{p}}$ onto their first $2l$ coordinates, which correspond to $A_{k+1},A_{k+2},\dots,A_{k+2l}$. Denote the projections by $\tilde{\mathbf{y}}$ and $\tilde{\mathbf{y}}_{\mathrm{p}}$. Then 
        \begin{equation}\label{eq:ytilde_even}
            \begin{aligned}
                \tilde{\mathbf{y}}&=t_1(q^{k+1}\mathbf{v}_1+q^{k+2l-1}\mathbf{v}_{2l-1})+\cdots+t_{l-1}(q^{k+l-1}\mathbf{v}_{l-1}+q^{k+l+1}\mathbf{v}_{l+1})+t_lq^{k+l}\mathbf{v}_l+\typ\\
                &=(\mathbf{v}_1,\cdots,\mathbf{v}_{2l})(t_1q^{k+1},\cdots,t_lq^{k+l},t_{l-1}q^{k+l+1},\cdots,t_1q^{k+2l-1},0)\trans+\typ\\
                &=M^{-1}\bm{\mu}+\typ,
            \end{aligned}
        \end{equation}
        where
        \[
        \bm{\mu}=(t_1q^{k+1},\cdots,t_lq^{k+l},t_{l-1}q^{k+l+1},\cdots,t_1q^{k+2l-1},0)\trans.
        \]

    \textbf{ Step 2:}    Now assume $n>2(l+1)q^2+2(l-1)$ and we deduce a contradiction. To rule out a defect-$l$ absolutely maximally entangled state in this range, it is enough to prove that for every $t_1,\dots,t_l\in\R$, the  vector $\mathbf{y}$ in Eq.~\eqref{sol_even} has a negative component. That is, some $A_i$ or $A_j'$ must be negative, which contradicts Lemma~\ref{lem:enumerator_constraints}. Suppose instead that some choice of $t_1,\dots,t_l$ makes all components of
        \[
        \mathbf{y}=t_1\mathbf{w}_1+\cdots+t_{l}\mathbf{w}_{l}+\yp
        \]
        nonnegative. Then $\tilde{\mathbf{y}}=M^{-1}\bm{\mu}+\typ$ is componentwise nonnegative as well.
        Assume that we have a vector
        \[
        \bm{\lambda}=(\lambda_1,\dots,\lambda_{2l})\trans\in\R_{\ge 0}^{2l}
        \]
        satisfying
        \begin{equation}\label{aim_even}
            \bm{\lambda}\trans M^{-1}\bm{\mu}=0.
        \end{equation} Then
  \[
        \bm{\lambda}\trans\tilde{\mathbf{y}}=\bm{\lambda}\trans(M^{-1}\bm{\mu}+\typ)=\bm{\lambda}\trans\typ\ge 0.
        \]
        In this way, we transfer the necessary condition $\mathbf{y}\ge 0$ for some $t_1,\dots,t_l\in\R$ to $\bm{\lambda}\trans\typ\geq 0$ for a specific $\bm{\lambda}=(\lambda_1,\dots,\lambda_{2l})\trans\in\R_{\ge 0}^{2l}$. If there exists such a $\bm{\lambda}$ satisfying Eq.~\eqref{aim_even} but $\bm{\lambda}\trans\typ<0$, then we have a contradiction. Next, we show such a  vector $\bm{\lambda}$ exists. 
      
        First, we show a  vector $\bm{\lambda}$ satisfying  Eq.~\eqref{aim_even} exists by presetting some of the entries. 

        \begin{claim}\label{claim_even}
            Set
            \begin{equation}\label{lambda_even}
                \lambda_{l+1}=l,\quad \lambda_{l+2}=\frac{l+1}{n-(k+l+1)},\quad \lambda_{l+3}=\dots=\lambda_{2l}=0.
            \end{equation}
            Then there exists $\lambda_1,\dots,\lambda_l\ge 0$ such that Eq.~\eqref{aim_even} holds.
        \end{claim}
        
        Assuming Claim~\ref{claim_even}, the preceding construction gives
        \[
        \bm{\lambda}\trans\tilde{\mathbf{y}}=\bm{\lambda}\trans\typ=\sum_{i=1}^{2l}\lambda_{i}a_{k+i}=la_{k+l+1}+\frac{l+1}{n-(k+l+1)}a_{k+l+2}\ge 0.
        \]
        Using Eq.~\eqref{value_even}, this becomes
        \[
        l(q^2-1)\binom{n}{k+l+1}+\frac{l+1}{n-(k+l+1)}\left((q^4-1)\binom{n}{k+l+2}-(n-(k+l+1))(q^2-1)\binom{n}{k+l+1}\right)\ge 0.
        \]
        Cancelling the common terms gives
        \[
        \frac{l+1}{n-(k+l+1)}(q^4-1)\binom{n}{k+l+2}\ge(q^2-1)\binom{n}{k+l+1}.
        \]
        Using the identity $\binom{n}{k+l+2}=\frac{n-(k+l+1)}{k+l+2}\binom{n}{k+l+1}$, we obtain
        \[\frac{l+1}{k+l+2}(q^2+1)\ge1, \qquad\text{i.e.}\qquad k\le(l+1)q^2-1.\]
        Consequently $n=2k+2l\le2(l+1)q^2+2(l-1)$, contradicting the assumed lower bound. It remains to prove Claim~\ref{claim_even}.

        We first record a binomial identity used in the verification of the claim, whose proof is given in the Appendix.

        \begin{lemma}\label{comeq}
            For $1\le j \le l$ and $s\ge l+2$, we have
            \begin{enumerate}[label=(\roman*)]
                \item $\displaystyle \sum_{i=j}^l(-1)^{l-i}\binom{s-i}{l+1-i}\binom{s-j}{i-j}=\binom{s-j}{l-j+1}$.
                \item $\displaystyle \sum_{i=j}^l(-1)^{l-i}\binom{s-i}{l+2-i}\binom{s-j}{i-j}=(l-j+1)\binom{s-j}{l-j+2}$.
            \end{enumerate}
        \end{lemma}

        We now verify Claim~\ref{claim_even}.

        \begin{proof}[Proof of Claim~\ref{claim_even}]
            We solve Eq.~\eqref{aim_even} with the prescribed tail values
            \[
            \lambda_{l+1}=l,\quad \lambda_{l+2}=\frac{l+1}{n-(k+l+1)},\quad \lambda_{l+3}=\dots=\lambda_{2l}=0,
            \]
            and then check that the resulting $\lambda_1,\dots,\lambda_l$ are nonnegative. We prove this in three steps.

            First, we solve $\bm{\lambda}\trans M^{-1}$ from $\bm{\lambda}\trans M^{-1}\bm{\mu}=0.$
            Define
            \[
            \bm{\lambda}\trans M^{-1}:=\bm{\alpha}\trans=(\alpha_1,\dots,\alpha_{2l}).
            \]
            Hence
            \[
            \bm{\lambda}\trans M^{-1}\bm{\mu}=\bm{\alpha}\trans\bm{\mu}=q^{k+1}t_1(\alpha_1+q^{2l-2}\alpha_{2l-1})+\cdots+q^{k+l-1}t_{l-1}(\alpha_{l-1}+q^2\alpha_{l+1})+q^{k+l}t_l\alpha_l.
            \]
            Since the parameters $t_1,\dots,t_l$ are arbitrary, the condition $\bm{\lambda}\trans M^{-1}\bm{\mu}=0$ is equivalent to
            \begin{equation}\label{eqsystem_even}
                \left\{
                \begin{aligned}
                    \alpha_{1}+q^{2l-2}\alpha_{2l-1} &= 0,\\
                    &\vdots\\
                    \alpha_{l-2}+q^4\alpha_{l+2} &= 0,\\
                    \alpha_{l-1}+q^2\alpha_{l+1} &= 0,\\
                    \alpha_l &= 0.
                \end{aligned}
                \right.
            \end{equation}
            To compute the $\alpha_i$, use the block decomposition
            \[
            M=\begin{pmatrix}
                M_1 & O\\
                M_2 & M_3
            \end{pmatrix},\quad
            M^{-1}=\begin{pmatrix}
                M_1' & O\\
                M_2' & M_3'
            \end{pmatrix},
            \]
            where each small block is of size $l\times l$. Put $s=n-k$. Then
            \begin{align}
                M_1'&=M_1^{-1},\label{eqm1_even}\\
                M_2'&=\begin{pmatrix}
                    (-1)^{l}\binom{s-1}{l} & (-1)^{l-1}\binom{s-2}{l-1} & \cdots & (-1)\binom{s-l}{1} \\
                    (-1)^{l+1}\binom{s-1}{l+1} & (-1)^{l}\binom{s-2}{l} & \cdots & (-1)^2\binom{s-l}{2} \\
                    \multicolumn{4}{c}{\vphantom{\begin{matrix}a_{11}\\a_{21}\end{matrix}} \ast}
                \end{pmatrix},\label{eqm2_even}\\
                M_3'&=\begin{pmatrix}
                    \binom{s-(l+1)}{0} & 0 & \cdots & 0 \\
                    -\binom{s-(l+1)}{1} & \binom{s-(l+2)}{0} & \cdots & 0 \\
                    \multicolumn{4}{c}{\vphantom{\begin{matrix}a_{11}\\a_{21}\end{matrix}} \ast}
                \end{pmatrix}.\label{eqm3_even}
            \end{align}
            With this decomposition, $\bm{\alpha}\trans=\bm{\lambda}\trans M^{-1}$ reads
            \[
            ((\alpha_1,\dots,\alpha_l),(\alpha_{l+1},\dots,\alpha_{2l}))=((\lambda_1,\dots,\lambda_l),(\lambda_{l+1},\cdots,\lambda_{2l}))\begin{pmatrix}
                M_1' & O\\
                M_2' & M_3'
            \end{pmatrix},
            \]
            Hence
            \begin{align}
                (\alpha_1,\dots,\alpha_l)&=(\lambda_1,\dots,\lambda_l)M_1'+(\lambda_{l+1},\cdots,\lambda_{2l})M_2',\label{eqa1_even}\\
                (\alpha_{l+1},\dots,\alpha_{2l})&=(\lambda_{l+1},\cdots,\lambda_{2l})M_3'.\label{eqa2_even}
            \end{align}
            Substitution of Eq.~\eqref{lambda_even} and Eq.~\eqref{eqm3_even} into Eq.~\eqref{eqa2_even} gives
            \begin{equation}\label{alpha2_even}
                \alpha_{l+1}=-1,\quad \alpha_{l+2}=\frac{l+1}{s-(l+1)},\quad \alpha_{l+3}=\dots=\alpha_{2l}=0.
            \end{equation}
            Then Eq.~\eqref{eqsystem_even} gives
            \begin{equation}\label{alpha1_even}
                \alpha_1=\dots=\alpha_{l-3}=0,\quad \alpha_{l-2}=-\frac{l+1}{s-(l+1)}q^4,\quad\alpha_{l-1}=q^2,\quad\alpha_l=0.
            \end{equation}

            Second, we present explicit formulas for $\lambda_1,\ldots, \lambda_l$.
            Set
            \[
            (\lambda_{l+1},\cdots,\lambda_{2l})M_2':=(\beta_1,\dots,\beta_l).
            \]
            Then
            \begin{equation}\label{beta_even}
                \begin{aligned}
                    \beta_i&=\lambda_{l+1}(M_2')_{1,i}+\lambda_{l+2}(M_2')_{2,i}\\
                    &=l(-1)^{l+1-i}\binom{s-i}{l+1-i}+\frac{l+1}{s-(l+1)}(-1)^{l-i}\binom{s-i}{l+2-i}.
                \end{aligned}
            \end{equation}
            Rewrite Eq.~\eqref{eqa1_even} as
            \begin{equation*}
                (\lambda_1,\dots,\lambda_l)M_1'=(\alpha_1,\dots,\alpha_l)-(\lambda_{l+1},\cdots,\lambda_{2l})M_2'=(\alpha_1-\beta_1,\cdots,\alpha_l-\beta_l).
            \end{equation*}
            Using Eq.~\eqref{eqm1_even}, this is equivalent to
            \begin{equation}\label{eqa1'_even}
                (\lambda_1,\dots,\lambda_l)=(\alpha_1-\beta_1,\cdots,\alpha_l-\beta_l)M_1.
            \end{equation}
            Using Eqs.~\eqref{alpha1_even}, \eqref{beta_even}, and \eqref{eqa1'_even}, we obtain for each $j\le l$
            \[
            \begin{aligned}
                \lambda_j&=\alpha_{l-2}M_{l-2,j}+\alpha_{l-1}M_{l-1,j}-\sum_{i=j}^{l}\beta_iM_{i,j}\\
                &=-\frac{l+1}{s-(l+1)}q^4\binom{s-j}{l-2-j}+q^2\binom{s-j}{l-1-j}+\sum_{i=j}^ll(-1)^{l-i}\binom{s-i}{l+1-i}\binom{s-j}{i-j}\\
                &\phantom{=}\,\,-\sum_{i=j}^l\frac{l+1}{s-(l+1)}(-1)^{l-i}\binom{s-i}{l+2-i}\binom{s-j}{i-j}\\
                &\xlongequal{\text{Lemma~\ref{comeq}}}-\frac{l+1}{s-(l+1)}q^4\binom{s-j}{l-j-2}+q^2\binom{s-j}{l-j-1}+l\binom{s-j}{l-j+1}-\frac{l+1}{s-(l+1)}(l-j+1)\binom{s-j}{l-j+2}\\
                &=-\frac{l+1}{s-(l+1)}q^4\binom{s-j}{l-j-2}+q^2\binom{s-j}{l-j-1}+(l-\frac{l+1}{s-(l+1)}(l-j+1)\frac{s-l-1}{l-j+2})\binom{s-j}{l-j+1}\\
                &=-\frac{l+1}{s-(l+1)}q^4\binom{s-j}{l-j-2}+q^2\binom{s-j}{l-j-1}+\frac{j-1}{l-j+2}\binom{s-j}{l-j+1}.
            \end{aligned}
            \]
            
           Finally, we check that all these coefficients are nonnegative.
            For $j=l$, $\lambda_l=q^2+(l-1)\binom{s-l}{2}>0$. Now assume $1\le j<l$. Since $n=2k+2l>2(l+1)q^2+2(l-1)$, we have $k>(l+1)q^2-1$. Denote
            \[
            m:=l-j-1\ (0\le m\le l-2), \quad a:=s-l=k+l>(l+1)q^2.
            \]
            Using the elementary ratios
            \[
            \begin{aligned}
                \binom{s-j}{l-j-2}&=\binom{s-j}{l-j-1}\frac{m}{a+2},\\
                \binom{s-j}{l-j+1}&=\binom{s-j}{l-j-1}\frac{(a+1)a}{(m+1)(m+2)},
            \end{aligned}
            \]
            we obtain
            \[
            \lambda_j=\binom{s-j}{l-j-1}(-q^4\frac{l+1}{a-1}\frac{m}{a+2}+q^2+\frac{j-1}{m+3}\frac{(a+1)a}{(m+1)(m+2)}),
            \]
            where
            \[
            0\le \frac{q^4(l+1)m}{(a-1)(a+2)}\le\frac{q^4(l+1)(l-2)}{(l+1)^2q^4}=\frac{l-2}{l+1}<1.
            \]
            Therefore
            \[
            \lambda_j\ge\binom{s-j}{l-j-1}(-1+q^2+0)\ge 0,
            \]
            which proves the claim.
        \end{proof}

        \subsection{\texorpdfstring{The Case of $n$ Odd}{The Case of n Odd}}
        Now let $n=2k+2l+1$. In this case, all the $4l$ variables in Eq.~\eqref{eqM}
        \[
        A_{k+1},A_{k+2},\dots,A_{k+2l},A_{k+1}',A_{k+2}',\dots,A_{k+2l}'
        \] are unknown. The proof follows the same strategy as the even case, but the parity changes the symmetry constraints. To make the comparison transparent, we keep the notation of the even case, with the necessary modifications.
        
        \textbf{Step 1:} First, 
        Eq.~\eqref{eqM} becomes
        \begin{equation}\label{eq_odd}
            L_{3l\times4l}\begin{pmatrix}
                A_{k+1} \\ \vdots \\ A_{k+2l} \\ A_{k+1}' \\ \vdots \\ A_{k+2l}'
            \end{pmatrix}=\begin{pmatrix}
                -\binom{n}{k+1} \\ \vdots \\ -\binom{n}{k+2l} \\ O_{l\times 1}
            \end{pmatrix},
        \end{equation}
        where
        \begin{equation*}
            L=\begin{pmatrix}
            M_{2l\times 2l} & P_{2l\times2l}\\
            O_{l\times 2l} & Q_{l\times 2l}
            \end{pmatrix}
        \end{equation*}
        with
        \begin{equation*}
            P=\begin{pmatrix}
                 -q^{k+1} &           &           &             &\\
              & -q^{k+2}  &           &             &\\
              &           & \ddots    &             &\\
              &           &           & -q^{k+2l} &\\           
            \end{pmatrix},\qquad Q=\begin{pmatrix}
            1        & \cdots & 0            & 0      & \cdots  & -1     \\
                     & \ddots & 0            & 0      & \iddots &        \\
                     &        & 1            & -1     &         &
            \end{pmatrix}.
        \end{equation*}

        In this case, $\rk L=3l$. Thus the solution set of Eq.~\eqref{eq_odd} is still an $l$-dimensional affine subspace in $\R^{4l}$. We  define $2l$ similar vectors in $\R^{4l}$:
        \[
        \mathbf{y}_1=\begin{pmatrix}
            q^{k+1}\mathbf{v}_1 \\ 1 \\ 0 \\ \vdots \\ 0
        \end{pmatrix},\qquad \mathbf{y}_2=\begin{pmatrix}
            q^{k+2}\mathbf{v}_2 \\ 0 \\ 1 \\ \vdots \\ 0
        \end{pmatrix},\qquad\cdots,\qquad \mathbf{y}_{2l}=\begin{pmatrix}
            q^{k+2l}\mathbf{v}_{2l} \\ 0 \\ 0 \\ \vdots \\ 1
        \end{pmatrix},
        \]
        which satisfy 
        $
            L(\mathbf{y}_1+\mathbf{y}_{2l})=L(\mathbf{y}_2+\mathbf{y}_{2l-1})=\cdots=L(\mathbf{y}_{l}+\mathbf{y}_{l+1})=0.
       $ This gives $l$  solutions to the homogeneous part of Eq. \eqref{eq_odd}:
        \[
        \mathbf{w}_1=\mathbf{y}_1+\mathbf{y}_{2l},\qquad\mathbf{w}_2=\mathbf{y}_2+\mathbf{y}_{2l-1},\qquad\cdots,\qquad \mathbf{w}_{l}=\mathbf{y}_{l}+\mathbf{y}_{l+1},
        \] which are also 
   linearly independent, and form a basis of $\mathrm{Ker}\,L$.

        The particular solution $\yp=(a_{k+1},\dots,a_{k+2l},a_{k+1}',\dots,a_{k+2l}')\trans$ to Eq.~\eqref{eq_odd} is chosen as before, which gives   
         \[
         \begin{aligned}
            a_{k+1}'=a_{k+2l}'&=\frac{1}{q^{k+1}}\binom{n}{k+1},\\
             &\vdots\\
            a_{k+l}'=a_{k+l+1}'&=\frac{1}{q^{k+l}}\binom{n}{k+l},
         \end{aligned}
         \]
         and
        \begin{equation}\label{value_odd}
            \begin{aligned}
                a_{k+l+1}&=(q-1)\binom{n}{k+l+1},\\
                a_{k+l+2}&=(q^3-1)\binom{n}{k+l+2}-(n-(k+l+1))(q-1)\binom{n}{k+l+1}.
            \end{aligned}
        \end{equation}
        Finally, each solution to 
        Eq.~\eqref{eq_odd} has the form
        \begin{equation}\label{sol_odd}
        \mathbf{y}=t_1\mathbf{w}_1+\cdots+t_{l}\mathbf{w}_{l}+\yp,\qquad t_1,\dots,t_l\in\R,
         \end{equation} and its projection to the first $2l$ coordinates satisfies
        \begin{equation*}
            \tilde{\mathbf{y}}=M^{-1}\bm{\mu}+\typ
        \end{equation*}
        with
        \[
        \bm{\mu}=(t_1q^{k+1},\cdots,t_lq^{k+l},t_{l}q^{k+l+1},\cdots,t_1q^{k+2l})\trans.
        \]

       \textbf{Step 2:} Now assume
        \[
        n>2(l+1)q^2+2(l+1)q+(2l-1).
        \]
        As in the even case, suppose for contradiction that some choice of $t_1,\dots,t_l$ makes all components of $\mathbf y$ in Eq.~\eqref{sol_odd} nonnegative. Then its projection
        \[
        \tilde{\mathbf y}=M^{-1}\bm{\mu}+\typ
        \]
        is also componentwise nonnegative. We will construct a vector
        \[
        \bm{\lambda}=(\lambda_1,\dots,\lambda_{2l})\trans\in\R_{\ge0}^{2l}
        \]
        satisfying
        \begin{equation}\label{aim_odd}
            \bm{\lambda}\trans M^{-1}\bm{\mu}=0.
        \end{equation}
        For such a vector, the assumed nonnegativity of $\tilde{\mathbf y}$ implies
        \[
        \bm{\lambda}\trans\tilde{\mathbf y}
        =
        \bm{\lambda}\trans(M^{-1}\bm{\mu}+\typ)
        =
        \bm{\lambda}\trans\typ
        \ge0.
        \]
        We will show that this necessary inequality contradicts the assumed lower bound on $n$.
       

        \begin{claim}\label{claim_odd}
            Set
            \begin{equation}\label{lambda_odd}
                \lambda_{l+1}=l,\quad \lambda_{l+2}=\frac{l+1}{n-(k+l+1)},\quad \lambda_{l+3}=\dots=\lambda_{2l}=0.
            \end{equation}
            Then there exists $\lambda_1,\dots,\lambda_l\ge 0$ such that Eq.~\eqref{aim_odd} holds.
        \end{claim}
       The proof of Claim~\ref{claim_odd} is similar to that of Claim~\ref{claim_even}, so we move it to the Appendix.
        Assuming Claim~\ref{claim_odd}, the preceding construction gives
        \[
        \bm{\lambda}\trans\tilde{\mathbf{y}}=\bm{\lambda}\trans\typ=\sum_{i=1}^{2l}\lambda_{i}a_{k+i}=la_{k+l+1}+\frac{l+1}{n-(k+l+1)}a_{k+l+2}\ge 0.
        \]
        Using Eq.~\eqref{value_odd}, this becomes
        \[
        l(q-1)\binom{n}{k+l+1}+\frac{l+1}{n-(k+l+1)}((q^3-1)\binom{n}{k+l+2}-(n-(k+l+1))(q-1)\binom{n}{k+l+1})\ge 0.
        \]
        After cancelling the common terms, we get
        \[
        \frac{l+1}{n-(k+l+1)}(q^3-1)\binom{n}{k+l+2}\ge(q-1)\binom{n}{k+l+1}.
        \]
        The identity $\binom{n}{k+l+2}=\frac{n-(k+l+1)}{k+l+2}\binom{n}{k+l+1}$ then gives
        \[\frac{l+1}{k+l+2}(q^2+q+1)\ge1, \qquad\text{i.e.}\qquad k\le(l+1)q^2+(l+1)q-1.\]
        Consequently $n=2k+2l+1\le2(l+1)q^2+2(l+1)q+(2l-1)$, contradicting the assumed lower bound. 

\section{\texorpdfstring{Asymptotic Upper Bound of $k$-Uniform States}{Asymptotic Upper Bound of k-Uniform States}}\label{sec:asymptotic}

    In this section, we apply Theorem~\ref{thmmain} to derive an explicit asymptotic upper bound on the ratio $k/n$ for $k$-uniform states in $\cqn$ when the local dimension $q$ is fixed.

    \begin{thm}\label{thm:asymptotic_kn_bound}
    Fix an integer $q\ge 2$. Suppose that there exists a $k$-uniform state in $(\C^q)^{\otimes n}$. Then
    \[
    \frac{k}{n}\le\begin{cases}
        \frac{q^2}{2(q^2+1)}+O(\frac{1}{n}), & n \text{ even},\\
        \frac{q^2+q}{2(q^2+q+1)}+O(\frac{1}{n}), & n \text{ odd}.
    \end{cases}
    \]
    In particular, one has the parity-free asymptotic bound
    \begin{equation}\label{eq:kn_unified}
    \limsup_{n\to\infty}\frac{k}{n}\le \frac{q^{2}+q}{2(q^{2}+q+1)}
    = \frac{1}{2}-\frac{1}{2(q^{2}+q+1)}.
    \end{equation}
    \end{thm}
    
    \begin{proof}
        Take the case of $n$ even as an example. Regard a $k$-uniform state as an AME state of defect $l$, where $l= n/2-k$. Then by Theorem~\ref{thmmain}, we have
        \[
        n\le 2(l+1)q^2+2(l-1), 
        \]
        i.e.
        \[
        n/2-k=l\ge\frac{n+4}{2(q^2+1)}-1.
        \]
        Therefore,
        \[
        \begin{aligned}
            k&\le n/2 -\frac{n+4}{2(q^2+1)}+1\\
            &\le \left(\frac{1}{2}-\frac{1}{2(q^2+1)}\right)n+\frac{q^2-1}{q^2+1}.
        \end{aligned}
        \]
        The case of $n$ odd is similar. Note that $\frac{q^2}{2(q^2+1)}\le\frac{q^2+q}{2(q^2+q+1)}$, hence we take the weaker one for the parity-free asymptotic bound:
        \[
        \limsup_{n\to\infty}\frac{k}{n}\le \frac{q^{2}+q}{2(q^{2}+q+1)} = \frac{1}{2}-\frac{1}{2(q^{2}+q+1)}.
        \]
    \end{proof}
    
    We next compare Eq.~\eqref{eq:kn_unified} with the asymptotic criterion of Ning et al.~\cite{NingShiLuoZhang2025}. Define, for $0\le \theta<1/2$,
    \begin{equation}\label{eq:ning_rate_function}
    \begin{aligned}
    \Phi_q(\theta)&\triangleq\frac{1}{2}\log(1-2\theta)
    +(1-\theta)\log\left(\frac{1-\theta}{1-2\theta}\right)
    +\theta\log 2-(1-2\theta)\log q.
    \end{aligned}
    \end{equation}
    Theorem~IV.5 of Ning et al.\ states that if $\Phi_q(\theta)>0$, then every $k$-uniform state in
    $(\C^q)^{\otimes n}$ satisfies $k\le \theta n$ for all sufficiently large $n$~\cite{NingShiLuoZhang2025}.
    Thus the strongest asymptotic constant obtainable from this criterion is
    \[
    \frac{k}{n}\leq \theta_{\mathrm N}(q):=\inf\{\theta\in[0,1/2):\Phi_q(\theta)>0\}.
    \]

      Next, we show that the parity-free asymptotic bound in Eq.~\eqref{eq:kn_unified} is strictly stronger than the
    bound obtained from the asymptotic criterion bound $\theta_{\mathrm N}(q)$ for every $q\ge4$.
    \begin{proposition}\label{prop:strict_improvement_ning}
   For every integer $q\ge2$,
    \[
    \frac{q^2}{2(q^2+1)}<\theta_{\mathrm N}(q).
    \]
   Moreover, for every integer $q\ge4$, 
    \[
    \frac{q^2+q}{2(q^2+q+1)}<\theta_{\mathrm N}(q).
    \]

    \end{proposition}

    \begin{proof}
    We first record the monotonicity range of $\Phi_q$. Differentiating Eq.~\eqref{eq:ning_rate_function} gives
    \[
    \Phi_q'(\theta)=\log\left(\frac{2q^2(1-2\theta)}{1-\theta}\right).
    \]
    Hence $\Phi_q$ is strictly increasing on
    \[
    \left[0,\frac{2q^2-1}{4q^2-1}\right).
    \]
    The two constants from Theorem~\ref{thm:asymptotic_kn_bound} lie in this interval, since
    \[
    \frac{2q^2-1}{4q^2-1}-\frac{q^2}{2(q^2+1)}
    =\frac{3q^2-2}{2(2q-1)(2q+1)(q^2+1)}>0,
    \]
    and
    \[
    \frac{2q^2-1}{4q^2-1}-\frac{q^2+q}{2(q^2+q+1)}
    =\frac{(q-1)(3q+2)}{2(2q-1)(2q+1)(q^2+q+1)}>0.
    \]

    For $R>0$, denote
    \[
    \theta_R:=\frac{R}{2(R+1)}.
    \]
    Since $1-2\theta_R=1/(R+1)$ and $1-\theta_R=(R+2)/(2(R+1))$, direct substitution gives
    \begin{equation}\label{eq:Phi_theta_R}
    \Phi_q(\theta_R)
    =\frac{(R+2)\log(R+2)-(R+1)\log(R+1)-2\log(2q)}{2(R+1)}.
    \end{equation}
Next, we evaluate Eq.~\eqref{eq:Phi_theta_R} at $R=q^2$ and $R=q^2+q$. We will prove that
\[
\Phi_q(\theta_{q^2})<0\quad (q\ge2),\qquad
\Phi_q(\theta_{q^2+q})<0\quad (q\ge4).
\]
Since $\Phi_q$ is strictly increasing on an interval containing both relevant points, these inequalities imply
\[
\frac{q^2}{2(q^2+1)}<\theta_{\mathrm N}(q)\quad (q\ge2),\qquad
\frac{q^2+q}{2(q^2+q+1)}<\theta_{\mathrm N}(q)\quad (q\ge4).
\]

    First take $R=q^2$. The numerator of Eq.~\eqref{eq:Phi_theta_R} is
\[
N_q=(q^2+2)\log(q^2+2)-(q^2+1)\log(q^2+1)-\log(4q^2).
\]
It is enough to show that $N_q<0$. Equivalently, set $x=q^2$ and define
\[
H(x)=\log(4x)-(x+2)\log(x+2)+(x+1)\log(x+1).
\]
Then $H(x)=-N_q$. Moreover,
\[
H'(x)=\frac{1}{x}-\log\left(1+\frac{1}{x+1}\right)>0
\]
for $x>0$. Since
\[
H(4)=\log\left(\frac{16\cdot5^5}{6^6}\right)>0,
\]
we have $H(q^2)>0$ for all $q\ge2$. Hence $N_q<0$, and therefore
\[
\Phi_q\left(\theta_{q^2}\right)<0 \quad \text{for } q\ge2.
\]

    Next take $R=q^2+q$. The numerator of Eq.~\eqref{eq:Phi_theta_R} is
\[
N_q'=(q^2+q+2)\log(q^2+q+2)
-(q^2+q+1)\log(q^2+q+1)-\log(4q^2).
\]
It is enough to show that $N_q'<0$. Equivalently, define
\[
K(q)=\log(4q^2)-(q^2+q+2)\log(q^2+q+2)
+(q^2+q+1)\log(q^2+q+1).
\]
Then $K(q)=-N_q'$. We have
\[
K(4)=\log\left(\frac{64\cdot21^{21}}{22^{22}}\right)>0.
\]
Moreover,
\[
\begin{aligned}
K'(q)
&=\frac{2}{q}-(2q+1)\log\left(1+\frac{1}{q^2+q+1}\right)\\
&>\frac{2}{q}-\frac{2q+1}{q^2+q+1}
=\frac{q+2}{q(q^2+q+1)}>0.
\end{aligned}
\]
Thus $K(q)>0$ for all $q\ge4$. Hence $N_q'<0$, and therefore
\[
\Phi_q\left(\theta_{q^2+q}\right)<0 \quad \text{for } q\ge4.
\]
    \end{proof}
Proposition~\ref{prop:strict_improvement_ning} shows that Theorem~\ref{thm:asymptotic_kn_bound} gives a strictly stronger asymptotic upper bound on the ratio $k/n$. A comparison for small $q\in\{4,5,6,7,8,9\}$ has been listed in Table~\ref{tab:compare}.
    
    

\section{Conclusion}

In this paper, we resolved the Scott-type nonexistence-bound conjecture posed by Ning et al.\ \cite{NingShiLuoZhang2025} for $k$-uniform states with arbitrary defect. Equivalently, the result gives explicit length bounds for one-dimensional pure quantum error-correcting codes near the quantum Singleton regime. The proof is based on a truncated MacWilliams linear-programming system for the quantum weight enumerator and the unitary enumerator, together with an explicit infeasibility certificate. This yields a closed-form Scott-type upper bound on the length $n$ in terms of $(q,l)$ for all $q\ge2$ and $l\in\mathbb Z_{\ge0}$.

As a direct application, the explicit bound gives new asymptotic upper bounds on the ratio $k/n$ for $k$-uniform states with fixed local dimension $q$. Compared with the work in~\cite{NingShiLuoZhang2025}, our results are explicit and stronger for every $q\ge4$. As shown in Table~\ref{tab:compare}, the present bound does not improve the shadow-inequality bounds for $q=2,3$~\cite{Rains1998Shadow,ShiNingZhaoZhang2025}. It would be interesting to explore whether shadow-inequality methods can yield comparable Scott-type bounds for general local dimension $q$.


\bibliography{refs}

@article{Scott2004,
  author  = {Scott, A. J.},
  title   = {Multipartite entanglement, quantum-error-correcting codes, and entangling power of quantum evolutions},
  journal = {Physical Review A},
  volume  = {69},
  number  = {5},
  pages   = {052330},
  year    = {2004},
  doi     = {10.1103/PhysRevA.69.052330}
}

@article{ShorLaflamme1997,
  author  = {Shor, Peter W. and Laflamme, Raymond},
  title   = {Quantum Analog of the {M}ac{W}illiams Identities for Classical Coding Theory},
  journal = {Physical Review Letters},
  volume  = {78},
  number  = {8},
  pages   = {1600--1602},
  year    = {1997},
  doi     = {10.1103/PhysRevLett.78.1600}
}

@article{CalderbankRainsShorSloane1998,
  author  = {Calderbank, A. R. and Rains, E. M. and Shor, P. M. and Sloane, N. J. A.},
  title   = {Quantum error correction via codes over {GF}(4)},
  journal = {IEEE Transactions on Information Theory},
  volume  = {44},
  number  = {4},
  pages   = {1369--1387},
  year    = {1998},
  doi     = {10.1109/18.681315}
}

@article{Rains1998WeightEnum,
  author  = {Rains, Eric M.},
  title   = {Quantum weight enumerators},
  journal = {IEEE Transactions on Information Theory},
  volume  = {44},
  number  = {4},
  pages   = {1388--1394},
  year    = {1998},
  doi     = {10.1109/18.681316}
}

@article{HuberEltschkaSiewertGuhne2018,
  author  = {Huber, Felix and Eltschka, Christopher and Siewert, Jens and G{\"u}hne, Otfried},
  title   = {Bounds on absolutely maximally entangled states from shadow inequalities, and the quantum {M}ac{W}illiams identity},
  journal = {Journal of Physics A: Mathematical and Theoretical},
  volume  = {51},
  number  = {17},
  pages   = {175301},
  year    = {2018},
  doi     = {10.1088/1751-8121/aab145}
}

@article{AshikhminLitsyn1999,
  author  = {Ashikhmin, Alexei and Litsyn, Simon},
  title   = {Upper bounds on the size of quantum codes},
  journal = {IEEE Transactions on Information Theory},
  volume  = {45},
  number  = {4},
  pages   = {1206--1215},
  year    = {1999},
  doi     = {10.1109/18.761271}
}

@article{GrasslHuberWinter2022,
  author  = {Grassl, Markus and Huber, Felix and Winter, Andreas},
  title   = {Entropic proofs of {S}ingleton bounds for quantum error-correcting codes},
  journal = {IEEE Transactions on Information Theory},
  volume  = {68},
  number  = {6},
  pages   = {3942--3950},
  year    = {2022},
  doi     = {10.1109/TIT.2022.3149291}
}

@misc{NingShiLuoZhang2025,
  author = {Ning, Yu and Shi, Fei and Luo, Tao and Zhang, Xiande},
  title  = {Linear Programming Bounds on $k$-Uniform States},
  year   = {2025},
  note   = {arXiv:2503.02222 [quant-ph]}
}

@article{HelwigCuiRieraLatorreLo2012,
  author  = {Helwig, Wolfram and Cui, Wei and Riera, Arnau and Latorre, Jos{\'e} I. and Lo, Hoi-Kwong},
  title   = {Absolute maximal entanglement and quantum secret sharing},
  journal = {Physical Review A},
  volume  = {86},
  number  = {5},
  pages   = {052335},
  year    = {2012},
  doi     = {10.1103/PhysRevA.86.052335}
}

@article{ShiLiChenZhang2021Masking,
  author  = {Shi, Fei and Li, Mao-Sheng and Chen, Lin and Zhang, Xiande},
  title   = {{$k$}-uniform quantum information masking},
  journal = {Physical Review A},
  volume  = {104},
  number  = {3},
  pages   = {032601},
  year    = {2021},
  doi     = {10.1103/PhysRevA.104.032601}
}

@article{GoyenecheZyczkowski2014,
  author  = {Goyeneche, Dardo and {\.Z}yczkowski, Karol},
  title   = {Genuinely multipartite entangled states and orthogonal arrays},
  journal = {Physical Review A},
  volume  = {90},
  number  = {2},
  pages   = {022316},
  year    = {2014},
  doi     = {10.1103/PhysRevA.90.022316}
}

@article{GoyenecheAlsinaLatorreRieraZyczkowski2015,
  author  = {Goyeneche, Dardo and Alsina, Daniel and Latorre, Jos{\'e} I. and Riera, Arnau and {\.Z}yczkowski, Karol},
  title   = {Absolutely maximally entangled states, combinatorial designs, and multiunitary matrices},
  journal = {Physical Review A},
  volume  = {92},
  number  = {3},
  pages   = {032316},
  year    = {2015},
  doi     = {10.1103/PhysRevA.92.032316}
}

@article{GoyenecheRaissiDiMartinoZyczkowski2018,
  author  = {Goyeneche, Dardo and Raissi, Zahra and Di Martino, Sara and {\.Z}yczkowski, Karol},
  title   = {Entanglement and quantum combinatorial designs},
  journal = {Physical Review A},
  volume  = {97},
  number  = {6},
  pages   = {062326},
  year    = {2018},
  doi     = {10.1103/PhysRevA.97.062326}
}

@article{RaissiTeixidoGogolinAcin2020,
  author  = {Raissi, Zahra and Teixid{\'o}, Adam and Gogolin, Christian and Ac{\'i}n, Antonio},
  title   = {Constructions of {$k$}-uniform and absolutely maximally entangled states beyond maximum distance codes},
  journal = {Physical Review Research},
  volume  = {2},
  number  = {3},
  pages   = {033411},
  year    = {2020},
  doi     = {10.1103/PhysRevResearch.2.033411}
}

@article{HuberGuhneSiewert2017,
  author  = {Huber, Felix and G{\"u}hne, Otfried and Siewert, Jens},
  title   = {Absolutely maximally entangled states of seven qubits do not exist},
  journal = {Physical Review Letters},
  volume  = {118},
  number  = {20},
  pages   = {200502},
  year    = {2017},
  doi     = {10.1103/PhysRevLett.118.200502}
}

@article{LiWang2019,
  author  = {Li, Mao-Sheng and Wang, Yan-Ling},
  title   = {{$k$}-uniform quantum states arising from orthogonal arrays},
  journal = {Physical Review A},
  volume  = {99},
  number  = {4},
  pages   = {042332},
  year    = {2019},
  doi     = {10.1103/PhysRevA.99.042332}
}

@article{FengJinXingYuan2017,
  author  = {Feng, Keqin and Jin, Lingfei and Xing, Chaoping and Yuan, Chen},
  title   = {Multipartite entangled states, symmetric matrices, and error-correcting codes},
  journal = {IEEE Transactions on Information Theory},
  volume  = {63},
  number  = {9},
  pages   = {5618--5627},
  year    = {2017},
  doi     = {10.1109/TIT.2017.2700866}
}

@article{Rains1999ShadowEnum,
  author  = {Rains, Eric M.},
  title   = {Quantum shadow enumerators},
  journal = {IEEE Transactions on Information Theory},
  volume  = {45},
  number  = {7},
  pages   = {2361--2366},
  year    = {1999},
  doi     = {10.1109/18.796376}
}

@article{ShiNingZhaoZhang2025,
  author  = {Shi, Fei and Ning, Yu and Zhao, Qi and Zhang, Xiande},
  title   = {Bounds on {$k$}-uniform quantum states},
  journal = {IEEE Transactions on Information Theory},
  volume  = {71},
  number  = {1},
  pages   = {413--425},
  year    = {2025},
  doi     = {10.1109/TIT.2024.3481042}
}

@article{Rains1998Shadow,
  author  = {Rains, Eric M.},
  title   = {Shadow bounds for self-dual codes},
  journal = {IEEE Transactions on Information Theory},
  volume  = {44},
  number  = {1},
  pages   = {134--139},
  year    = {1998},
  doi     = {10.1109/18.651007}
}

@article{Han2006Nonexistence,
  author  = {Han, Sangwook and Lee, Jon-Lark Kim},
  title   = {Nonexistence of some extremal self-dual codes},
  journal = {Journal of the Korean Mathematical Society},
  volume  = {43},
  number  = {6},
  pages   = {1357--1369},
  year    = {2006}
}

@article{Han2009Nonexistence,
  author  = {Han, Sangwook and Kim, Jon-Lark},
  title   = {The nonexistence of near-extremal formally self-dual codes},
  journal = {Designs, Codes and Cryptography},
  volume  = {51},
  number  = {1},
  pages   = {69--77},
  year    = {2009},
  doi     = {10.1007/s10623-008-9247-8}
}

@misc{Huber:AMEtable,
  author = {Huber, Felix and Wyderka, Nicolai},
  title  = {Table of {AME} states},
  note   = {Online available at \url{https://tp.nt.uni-siegen.de/ame/ame.html}, accessed on 2026-05-13}
}

@book{NielsenChuang2010,
  author    = {Nielsen, Michael A. and Chuang, Isaac L.},
  title     = {Quantum Computation and Quantum Information},
  edition   = {10th Anniversary},
  publisher = {Cambridge University Press},
  address   = {Cambridge},
  year      = {2010}
}

@book{Watrous2018,
  author    = {Watrous, John},
  title     = {The Theory of Quantum Information},
  publisher = {Cambridge University Press},
  address   = {Cambridge},
  year      = {2018}
}
\appendix
\section{Appendix}\label{Appendix}
    \begin{proof}[Proof of Lemma~\ref{comeq}]
           (1) Observe that
                \[
                \begin{aligned}
                    \binom{s-i}{l+1-i}\binom{s-j}{i-j}&=\frac{(s-i)!}{(l+1-i)!(s-l-1)!}\frac{(s-j)!}{(i-j)!(s-i)!}\\
                    &=\frac{(s-j)!}{(s-l-1)!(l-j+1)!}\frac{(l-j+1)!}{(i-j)!(l+1-i)!}\\
                    &=\binom{s-j}{l-j+1}\binom{l-j+1}{l-i+1}.
                \end{aligned}
                \]
                Hence
                \[
                \begin{aligned}
                    \sum_{i=j}^l(-1)^{l-i}\binom{s-i}{l+1-i}\binom{s-j}{i-j}&=\binom{s-j}{l-j+1}\sum_{i=j}^l(-1)^{l-i}\binom{l-j+1}{l-i+1}\\
                    &\xlongequal{k=l-i+1}\binom{s-j}{l-j+1}\sum_{k=1}^{l-j+1}(-1)^{k-1}\binom{l-j+1}{k}\\
                    &=\binom{s-j}{l-j+1}.
                \end{aligned}
                \]
               (2) The same factorial rearrangement gives
                \[
                \binom{s-i}{l+2-i}\binom{s-j}{i-j}=\binom{s-j}{l-j+2}\binom{l-j+2}{l-i+2}.
                \]
                Hence
                \[
                \begin{aligned}
                    \sum_{i=j}^l(-1)^{l-i}\binom{s-i}{l+2-i}\binom{s-j}{i-j}&=\binom{s-j}{l-j+2}\sum_{i=j}^l(-1)^{l-i}\binom{l-j+2}{l-i+2}\\
                    &\xlongequal{k=l-i+2}\binom{s-j}{l-j+2}\sum_{k=2}^{l-j+2}(-1)^{k}\binom{l-j+2}{k}\\
                    &=\binom{s-j}{l-j+2}(-1+(l-j+2))\\
                    &=(l-j+1)\binom{s-j}{l-j+2}.
                \end{aligned}
                \]
            
        \end{proof}

        \begin{proof}[Proof of Claim~\ref{claim_odd}]
            We solve Eq.~\eqref{aim_odd} with the prescribed tail values
            \[
            \lambda_{l+1}=l,\quad \lambda_{l+2}=\frac{l+1}{n-(k+l+1)},\quad \lambda_{l+3}=\dots=\lambda_{2l}=0,
            \]
            and then verify that the resulting $\lambda_1,\dots,\lambda_l$ are nonnegative. Set
            \[
            \bm{\lambda}\trans M^{-1}:=\bm{\alpha}\trans=(\alpha_1,\dots,\alpha_{2l}).
            \]
            Hence
            \[
            \bm{\lambda}\trans M^{-1}\bm{\mu}=\bm{\alpha}\trans\bm{\mu}=q^{k+1}t_1(\alpha_1+q^{2l-1}\alpha_{2l})+q^{k+2}t_2(\alpha_2+q^{2l-3}\alpha_{2l-1})\cdots+q^{k+l}t_{l}(\alpha_{l}+q\alpha_{l+1}).
            \]
            Since the parameters $t_1,\dots,t_l$ are arbitrary, the condition $\bm{\lambda}\trans M^{-1}\bm{\mu}=0$ is equivalent to
            \begin{equation}\label{eqsystem_odd}
                \left\{
                \begin{aligned}
                    \alpha_{1}+q^{2l-1}\alpha_{2l} &= 0,\\
                    &\vdots\\
                    \alpha_{l-1}+q^3\alpha_{l+2} &= 0,\\
                    \alpha_{l}+q\alpha_{l+1} &= 0.
                \end{aligned}
                \right.
            \end{equation}
            To compute the $\alpha_i$, put $s=n-k$ and use the same block decomposition of $M^{-1}$:
            \[
            ((\alpha_1,\dots,\alpha_l),(\alpha_{l+1},\dots,\alpha_{2l}))=((\lambda_1,\dots,\lambda_l),(\lambda_{l+1},\cdots,\lambda_{2l}))\begin{pmatrix}
                M_1' & O\\
                M_2' & M_3'
            \end{pmatrix},
            \]
            Hence
            \begin{align}
                (\alpha_1,\dots,\alpha_l)&=(\lambda_1,\dots,\lambda_l)M_1'+(\lambda_{l+1},\cdots,\lambda_{2l})M_2',\label{eqa1_odd}\\
                (\alpha_{l+1},\dots,\alpha_{2l})&=(\lambda_{l+1},\cdots,\lambda_{2l})M_3'.\label{eqa2_odd}
            \end{align}
            Substitution of Eq.~\eqref{lambda_odd} and Eq.~\eqref{eqm3_even} into Eq.~\eqref{eqa2_odd} gives
            \begin{equation}\label{alpha2_odd}
                \alpha_{l+1}=-1,\quad \alpha_{l+2}=\frac{l+1}{s-(l+1)},\quad \alpha_{l+3}=\dots=\alpha_{2l}=0.
            \end{equation}
            Then Eq.~\eqref{eqsystem_odd} gives
            \begin{equation}\label{alpha1_odd}
                \alpha_1=\dots=\alpha_{l-2}=0,\quad \alpha_{l-1}=-\frac{l+1}{s-(l+1)}q^3,\quad\alpha_{l}=q.
            \end{equation}
            Set
            \[
            (\lambda_{l+1},\cdots,\lambda_{2l})M_2'=(\beta_1,\dots,\beta_l).
            \]
            Then
            \begin{equation}\label{beta_odd}
                \begin{aligned}
                    \beta_i&=\lambda_{l+1}(M_2')_{1,i}+\lambda_{l+2}(M_2')_{2,i}\\
                    &=l(-1)^{l+1-i}\binom{s-i}{l+1-i}+\frac{l+1}{s-(l+1)}(-1)^{l-i}\binom{s-i}{l+2-i}.
                \end{aligned}
            \end{equation}
            Rewrite Eq.~\eqref{eqa1_odd} as
            \begin{equation*}
                (\lambda_1,\dots,\lambda_l)M_1'=(\alpha_1,\dots,\alpha_l)-(\lambda_{l+1},\cdots,\lambda_{2l})M_2'=(\alpha_1-\beta_1,\cdots,\alpha_l-\beta_l).
            \end{equation*}
            Using Eq.~\eqref{eqm1_even}, this is equivalent to
            \begin{equation}\label{eqa1'_odd}
                (\lambda_1,\dots,\lambda_l)=(\alpha_1-\beta_1,\cdots,\alpha_l-\beta_l)M_1
            \end{equation}
            Combining Eqs.~\eqref{alpha1_odd}, \eqref{beta_odd}, and \eqref{eqa1'_odd}, we obtain
            \[
            \begin{aligned}
                \lambda_j&=\alpha_{l-1}M_{l-1,j}+\alpha_{l}M_{l,j}-\sum_{i=j}^{l}\beta_iM_{i,j}\\
                &=-\frac{l+1}{s-(l+1)}q^3\binom{s-j}{l-1-j}+q\binom{s-j}{l-j}+\sum_{i=j}^ll(-1)^{l-i}\binom{s-i}{l+1-i}\binom{s-j}{i-j}\\
                &\phantom{=}\,\,-\sum_{i=j}^l\frac{l+1}{s-(l+1)}(-1)^{l-i}\binom{s-i}{l+2-i}\binom{s-j}{i-j}\\
                &\xlongequal{\text{Lemma~\ref{comeq}}}-\frac{l+1}{s-(l+1)}q^3\binom{s-j}{l-j-1}+q\binom{s-j}{l-j}+l\binom{s-j}{l-j+1}-\frac{l+1}{s-(l+1)}(l-j+1)\binom{s-j}{l-j+2}\\
                &=-\frac{l+1}{s-(l+1)}q^3\binom{s-j}{l-j-1}+q\binom{s-j}{l-j}+(l-\frac{l+1}{s-(l+1)}(l-j+1)\frac{s-l-1}{l-j+2})\binom{s-j}{l-j+1}\\
                &=-\frac{l+1}{s-(l+1)}q^3\binom{s-j}{l-j-1}+q\binom{s-j}{l-j}+\frac{j-1}{l-j+2}\binom{s-j}{l-j+1}.
            \end{aligned}
            \]
            
            It remains only to check that these coefficients are nonnegative.

            Since $n=2k+2l+1>2(l+1)q^2+2(l+1)q+(2l-1)$, we have $k>(l+1)q^2+(l+1)q-1$. Denote
            \[
            m:=l-j\ (0\le m\le l-1), \quad a:=s-l=k+l+1>(l+1)q^2+1.
            \]
            The elementary ratios
            \[
            \begin{aligned}
                \binom{s-j}{l-j-1}&=\binom{s-j}{l-j}\frac{m}{a+1},\\
                \binom{s-j}{l-j+1}&=\binom{s-j}{l-j}\frac{a}{m+1},
            \end{aligned}
            \]
            give
            \[
            \lambda_j=\binom{s-j}{l-j}(-q^3\frac{l+1}{a-1}\frac{m}{a+1}+q+\frac{j-1}{m+2}\frac{a}{m+1}),
            \]
            where
            \[
            0\le \frac{q^3(l+1)m}{(a-1)(a+1)}\le\frac{q^3(l+1)(l-1)}{(l+1)^2q^4}=\frac{l-1}{q(l+1)}<1.
            \]
            It follows that
            \[
            \lambda_j\ge\binom{s-j}{l-j}(-1+q+0)\ge 0,
            \]
            and this proves the claim.
        \end{proof}

\end{document}